\documentclass{eptcs}
\providecommand{\event}{SOS 2007} 
\usepackage{breakurl}             

\usepackage{amsmath,amsthm,amsfonts,amssymb,amscd}
\usepackage{pstricks}
\usepackage[latin1]{inputenc}

\newcommand{\kom}[1]{}
\newcommand{\abet}{\mathcal{A}}
\newcommand{\bbet}{\mathcal{B}}
\newcommand{\M}{\mathcal{M}}

\newtheorem{thm}{Theorem}

\theoremstyle{definition}

\title{A new proof for the decidability of D0L ultimate periodicity}

\author{Vesa Halava$^{1}$, Tero Harju$^{1}$and Tomi K\"arki$^{1,2}$ \\ \\
\institute{${}^{1}$Department of Mathematics\\
       University of~Turku, 20014 Turku, Finland}
\institute{${}^{2}$Department of Teacher Education\\
		University of~Turku, PO Box 175, 26101 Rauma, Finland}
}

\def\titlerunning{D0L ultimate periodicity}
\def\authorrunning{V. Halava, T. Harju \& T. Kärki}
\begin{document}
\maketitle

\begin{abstract}
We give a new proof for the decidability of the D0L ultimate
periodicity problem based on the decidability of $p$-periodicity of
morphic words adapted to the approach of Harju and Linna.
\end{abstract}

\section{Introduction}

L systems were originally introduced by A. Lindenmayer to model the
development of simple filamentous organisms~\cite{Lin68a, Lin68b}. The
challenging and fruitful study of these systems in the 70s and 80s
created many new results and notions~\cite{RS86}. In this paper we
consider the important problem of recognizing ultimately periodic
D0L sequences.

Let $\abet$ be a finite alphabet and denote the empty word by
$\varepsilon$. A \emph{D0L system} is a pair $(h,u)$, where $h \colon
\abet^* \to \abet^*$ is a morphism and $u$ is a finite word over
$\abet$. The \emph{language} of the D0L system is $L(h,u)=\{h^i(u)
\mid i \geq 0\}$ and the \emph{limit set} $\lim L(h,u)$ consists of all
infinite words $w$ such that for all $n$ there is a prefix of $w$
longer than $n$ belonging to $L(h,u)$. Clearly, if the limit set is
non-empty, then one can effectively find integers $p$ and $q$ such
that $h^p(u)$ is a proper prefix of $h^{p+q}(u)$ and
\[
\lim L(h,u) = \bigcup_{i=0}^{q-1} \lim L(h^q, h^{p+i}(u)),
\]
where $|\lim L(h^q, h^{p+i}(u))|=1$. Hence, we may restrict to D0L
systems $(h,u)$ where $h$ is prolongable on $u$, i.e., $h(u)=uy$ and
$h^n(y)\neq \varepsilon$ for all integers $n\geq0$. In this case,
$h^n(u)$ is a prefix of $h^{n+1}(u)$ and the limit is the following fixed
point of $h$:
\[
h^\omega(u)=\lim_{n \to \infty}h^n(u)=uyh(y)h^2(y)\cdots.
\]

An infinite word $x$ is ultimately periodic if it is of the form
$x=uv^\omega=uvvv\cdots$, where $u$ and $v$ are finite words. The
length $|u|$ is a \emph{preperiod} and the length $|v|$ is a
\emph{period} of~$x$. An infinite word $x$ is \emph{ultimately
$p$-periodic} if $|v|=p$. The smallest period of $x$ is called
\emph{the period} of~$x$.

Now we are ready to formulate the \emph{D0L ultimate periodicity
problem}: \emph{Given a morphism $h$ prolongable on $u$, decide
whether $h^\omega(u)$ is ultimately periodic.} Note that in this
problem we may assume that $u$ is a letter. Indeed, if $h(u)=uy$,
then instead of $(h,u)$ we may consider $(h',a)$ where $a \not \in
\abet$ and $h'\colon (\abet \cup \{a\})^* \to (\abet \cup \{a\})^*$
where $h'(a)=ay$ and $h'(b)=h(b)$ for every $b \in \abet$. The limit
$h^\omega(u)$ is ultimately periodic if and only if $h'^\omega(a)$ is.

The decidability of the ultimate periodicity question for D0L
sequences was proven by T.~Harju and M.~Linna~\cite{HarLin86} and,
independently, by J.-J. Pansiot~\cite{Pan86}; see also a more recent
proof of J. Honkala~\cite{Hon08}. In the binary case the problem was
effectively solved by Séébold~\cite{See88}. Here we show how the
proof of~\cite{HarLin86} can be simplified using a recent result
concerning the decidability of the $p$-periodicity problem.

Before giving the proof, we introduce the following notation. Given a
morphism $h\colon \abet^* \to \abet^*$, we call a letter $b \in \abet$
\emph{finite} \kom{(with respect to $h$)} if $\{h^n(b) \mid n \geq
0\}$ is a finite set. Otherwise, $b$ is an \emph{infinite letter}.
Moreover, we say that a letter $b$ is \emph{recurrent} in
$h^\omega(a)$ if it occurs infinitely often in~$h^\omega(a)$. For a
given morphism $h$ prolongable on $a$ and for an infinite word
$h^\omega(a)$, denote the set of finite letters by $\abet_F$, the set
on infinite letters by $\abet_I$ and the set of recurrent letters by
$\abet_R$. Also, denote by $\abet_1$ the subset of $\abet$ which
consists of the infinite letters occurring infinitely many times
in~$h^\omega(a)$, i.e., $\abet_1=\abet_I \cap \abet_R$.

Let us shortly describe how the sets $\abet_F$, $\abet_I$ and
$\abet_R$ can be constructed. Note that if $b$ is a mortal letter, i.e.,
$h^n(b)=\varepsilon$ for some $n \geq 1$, then
$h^{|\abet|}(b)=\varepsilon$. Denote $\hat{h}=h^{|\abet|}$ and
denote the set of the mortal letters by $\M$. Note also that $b$ is a finite
letter if and only if there exists a word $u \in \{h^n(b) \mid n \geq 0\}$
such that $u = h^p(u)$ for some $p \geq 1$. Clearly, $\{\hat{h}^n(b)
\mid n \geq 0 \}$ is finite if and only if $\{h^n(b) \mid n \geq 0 \}$ is
finite. Hence, by replacing $h$ with $\hat{h}$ we may assume that
$h(b)=\varepsilon$ if $b \in \M$. Moreover, let $\bbet=\abet \setminus
\M$ and let $g \colon \bbet^* \to \bbet^*$ be a morphism defined by
$g(b)=\mu h(b)$, where
\[
\mu(b)=\left \{ \begin{array}{rr}
\varepsilon, & \text{if $b \in \M$,}\\
b, & \text{otherwise.}
\end{array}
\right.
\]
Now $g$ is non-erasing, and $b \in \abet_F$ if and only if $\{g^n(b)
\mid n \geq 0\}$ is finite. Namely, for any $n \geq 0$, we know by the
definition of $g$ that the word $h^n(b)$ can be obtained by inserting
a finite number of mortal letters to $g^n(b)$. The set $\{g^n(b) \mid n
\geq 0\}$ is finite if and only if for some $n$ all letters in $g^n(b)$
belong to $U_1=\{b \in\bbet \mid g^i(b) \in \bbet \text{ for every
$i\geq 0$}\}$. If $U_i=\{b \in \bbet \mid g(b) \in U_{i-1}^*\}$, then
$U_{i-1} \subseteq U_i$ and
\[
\abet_F \setminus \M = \bigcup_{i=1}^\infty U_i= U_{|\abet|}.
\]
Hence, we can effectively calculate $\abet_F$ and $\abet_I=\abet
\setminus \abet_F$. In order to find the recursive letters, we
construct a graph $G$ where the set of vertices is $\abet$ and there
is an edge from $b$ to $c$ if $c$ occurs in the image $h(b)$. Let $h(a)=ax$. If there
are infinitely many paths from a letter in~$x$ to the letter~$b$,
then $b$ occurs infinitely many times in $h^\omega(a)$.

\goodbreak
\section{Decidability of the $p$-periodicity problem}

Let $p \geq 1$, and let $x=(x_n)_{n \geq 0}$ be an infinite word over $\abet=\{a_1, \ldots,
a_d\}$. For $0 \leq k \leq p-1$, we say that the letters occurring
infinitely many times in positions~$x_n$, where  $n \equiv k \pmod{p}$,
form the \emph{$k$-set of $x$ modulo $p$}. It was shown in~\cite{HHKR10}
that these $k$-sets can be effectively constructed for $x = h^\omega(u)$, where $h$ is
prolongable on the word $u$. This is based
on the fact that there exist integers $r$ and $q$ such that
\begin{equation}\label{eqLemma}
|h^r(b)| \equiv |h^{r+q}(b)| \pmod{p}
\end{equation}
for every letter $b \in \abet$. The incidence matrix of $h$ is the matrix
$M=(m_{i,j})_{1 \leq i,j \leq d}$ where $m_{i,j}$ denotes the number
of occurrences of $a_i$ in $h(a_j)$. The sequence of matrices
$M^n \bmod{p}$, where the entries are the residues modulo~$p$,
must be ultimately periodic. Since $|h^n(a_j)| \pmod{p}$ is the sum of
the elements in the $j$th column of $M^n$, we conclude that the
sequence $(|h^n(a_j)|)_{n \geq 0}  \pmod{p}$ is ultimately periodic for
every $a_j \in \abet$ and \eqref{eqLemma} follows.

In order to find the $k$-sets of $x$ modulo $p$ we construct a
directed graph $G_{h}=(V,E)$ where the set of vertices
$V$ is $\{ (a,i) \mid a \in \abet, \  0\leq i <p\}$ and there is an edge
from $(c,i)$ to $(d,j)$ if, for some $b$ in $x$, the letter $c$ occurs in the image $h^r(b)$ at position congruent to $i$ (mod $p$) in $x$, and the letter $d$ occurs in the image $h^q(c)$ at position congruent to $j$ (mod $p$) in $x$; see Figure~\ref{fig1}.

\begin{figure}
\centering
\begin{pspicture}(12,5)(0,0)
\psset{unit=0.5cm}
\psset{fillcolor=lightgray}
\psset{dimen=middle}
\psset{tbarsize=7pt 3}

\put(0.2,1.7){$x_0 \cdots x_{l-1} \mathbf{b}$} 
\psline{|*-}(0,1.5)(4,1.5)
\psline{|*-}(4,1.5)(4.7,1.5)
\put(4.85,1.3){$\cdots$}
\psline{-|*}(5.8,1.5)(6.6,1.5)
\put(6.8,1.7){$y_1 \cdots y_{m-1} \ \mathbf{c} \ y_{m+1} \cdots y_n$}
\psline{-|*}(6.6,1.5)(14.3,1.5)
\psline{-}(14.3,1.5)(15,1.5)
\put(15.15,1.3){$\cdots$}
\psline{-|*}(16.1,1.5)(16.9,1.5)
\psline{-}(16.9,1.5)(17.6,1.5)
\put(17.75,1.3){$\cdots$}
\psline{-|*}(18.7,1.5)(19.5,1.5)
\psline{-|*}(19.5,1.5)(22,1.5)
\psline{-}(22,1.5)(22.8,1.5)
\put(22.9,1.3){$\cdots$}
\put(20.5,1.7){$\mathbf{d}$}

\put(0,2.5){$\overbrace{\hspace{8.4cm}}^\text{\normalsize $h^{r+q}(x_0 \cdots x_{l-1})$}$}
\put(16.8,2.5){$\overbrace{\hspace{3.4cm}}^\text{\normalsize $h^{r+q}(\mathbf{b})$}$}

\put(14.9,0.7){$\underbrace{\hspace{1.3cm}}_\text{\normalsize $h^{q}(y_1 \cdots y_{m-1})\rule{1cm}{0cm}$}$}
\put(0,0.7){$\underbrace{\hspace{3.3cm}}_\text{\normalsize $h^{r}(x_0 \cdots x_{l-1})$}$}
\put(6.7,0.7){$\underbrace{\hspace{3.8cm}}_\text{\normalsize $h^{r}(\mathbf{b})$}$}
\put(19.5,0.7){$\underbrace{\hspace{1.25cm}}_\text{\normalsize $h^{q}(\mathbf{c})$}$}

\end{pspicture}
\caption{Images $h^r(b)$ and $h^{r+q}(b)$.}\label{fig1}
\end{figure}
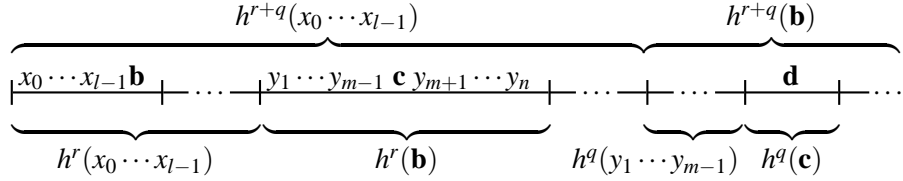

It is possible to construct such a graph by calculating the images
$h^r(b)$ and $h^{r+q}(b)$ for every $b \in \abet$. Namely, if $b=x_l$ and $c$ is the $m$th letter of $h^r(b)=y_1
\cdots y_n$ and $d$ is the $m'$th letter of $h^q(c)$, then we have
\begin{eqnarray}
i &\equiv& |h^r(x_0 \cdots x_{l-1})|+m-1 \pmod p,\label{eqi}\\
j &\equiv& |h^{r+q}(x_0 \cdots x_{l-1})|+|h^{q}(y_1 \cdots y_{m-1})|+m' -1 \pmod p.\label{eqj}
\end{eqnarray}
By~\eqref{eqLemma}, we have $|h^{r+q}(x_0 \cdots x_{l-1})|  \equiv
|h^r(x_0 \cdots x_{l-1})|$ (mod $p$), which together
with~\eqref{eqi} and~\eqref{eqj} implies
\[
j \equiv |h^{q}(y_1 \cdots y_{m-1})|+i+m'-m \pmod p.
\]
We say that a vertex $(c,i) \in V$ is an \emph{initial vertex} if there
exists a letter $b=x_l$ such that $0 \leq l < |h^r(a)|$, $c$ is the $m$th
letter of $h^r(b)$ and $i$ satisfies~\eqref{eqi}. A vertex $(c,k)$ is
called \emph{recurrent} if there exist infinitely many paths starting
from some initial vertex and ending in~$(c,k)$. By construction, this
means that $c$ belongs to the $k$ set of~$x$ modulo~$p$.

Given a coding $g$ and a morphism $h\colon \abet^* \to \abet^*$
prolongable on $a$, it is easy to see that the morphic word
$g(h^\omega(a))$ is ultimately $p$-periodic if and only if $g(b)=g(c)$
for all pairs of letters $(b,c)$ such that $b$ and $c$ belong to the
same $k$-set of~$h^\omega(a)$ modulo~$p$. Since the $k$-sets of
$h^\omega(a)$ can be effectively constructed, we have the following
result proved in~\cite{HHKR10}.

\begin{thm}\label{thPperiodic}
Given a positive integer $p$, it is decidable whether a morphic word
$g(h^\omega(a))$ is ultimately $p$-periodic.
\end{thm}

\section{Decidability of the D0L ultimate periodicity problem}

Before the decidability proof, we give the following result proved
in~\cite{Cul78, ER78}; see also~\cite{Hon08}.

\begin{thm}\label{thCancellation}
Let $h\colon \abet^* \to \abet^*$ be a morphism and $u, v \in
\abet^*$. If there is a positive integer $n$ such that $h^n(u)=h^n(v)$,
then $h^{|\abet|}(u)=h^{|\abet|}(v)$.
\end{thm}

This theorem can be proved by induction on the size of the alphabet
and the induction step is based on elementary morphisms. A
morphism $h \colon \abet^* \to \abet^*$ is called \emph{elementary}
if there do not exist an alphabet $\bbet$ smaller than $\abet$ and two
morphisms $f\colon \abet^* \to \bbet^*$ and $g\colon \bbet^* \to
\abet^*$ such that $h=gf$. Since elementary morphisms are injective,
the claim is clear if $h$ is elementary. Now assume that $h=gf$ as
above. Then $h^n(u)=h^n(v)$ implies that $(fg)^nf(u)=(fg)^nf(v)$ and,
by induction, $(fg)^{|\bbet|}f(u)=(fg)^{|\bbet|}f(v)$. This proves the
claim, since $(gf)^{|\bbet|+1}(u)=(gf)^{|\bbet|+1}(v)$ and
$|\abet|\geq |\bbet|+1$.

Using Theorem~\ref{thPperiodic} and Theorem~\ref{thCancellation}
and following the guidelines in~\cite{HarLin86} we give a new proof
for the decidability of the D0L ultimate periodicity problem.
The difference between the original proof of Harju and Linna and this proof is that we employ a new method obtained from $p$-periodicity as stated in Theorem~\ref{thPperiodic}.

\begin{thm}
The ultimate periodicity problem is decidable for D0L sequences.
\end{thm}

\begin{proof}
As explained above, it suffices to show that we can decide whether
$h^\omega(a)$ is ultimately periodic for a given morphism $h\colon
\abet^* \to \abet^*$ prolongable on $a$. Without loss of generality,
we assume that every letter of $\abet$ really occurs in
$h^\omega(a)$. Otherwise, we could consider a restriction of $h$.
Recall also that $\abet_1$ is the subset of $\abet$ which consists of
the infinite letters occurring infinitely many times in~$h^\omega(a)$.

If $\abet_1=\emptyset$, then the sequence is ultimately periodic.
Namely, if $h(a)=ay$ and $y$ contains infinite letters, then every
image $h^n(y)$ contains infinite letters and there must be at least one
infinite letter occurring infinitely many times in
$h^\omega(a)=ayh(y)h^2(y)\cdots$, which means that $\abet_1 \neq
\emptyset$. Therefore, there is only one infinite letter and it is the
letter $a$ occurring once in the beginning of the word. Hence,
$h(a)=ay$ where $y$ consists of finite letters. Then there must be
integers $n$ and $p$ such that $h^{n+p}(y)=h^{n}(y)$. Thus
$|h^n(y)h^{n+1}(y) \cdots h^{n+p-1}(y)|$ is a period of
$h^\omega(a)$.

Assume now that $b \in \abet_1$. We may write
\[
h^\omega(a)=u_0bu_1bu_2 \cdots,
\]
where $u_i \in (\abet \setminus \{b\})^*$. If the set $U=\{u_i \mid i
\geq 0\}$ is infinite then $h^\omega(a)$ cannot be ultimately
periodic. Note that if there exists a $c \in \abet_I$ such that the letter
$b$ does not occur in any $h^i(c)$, then $U$ is infinite. This property
is clearly decidable since if a letter occurs in $h^i(c)$ for some $i$,
then it occurs in the image for $i \leq |\abet|$. Hence, we may
assume that for each infinite letter $c$ the letter $b$ occurs in
$h^i(c)$ for some $i \leq |\abet|$.

Next we show that we may decide if $U$ is infinite or not. First
assume that $U$ is infinite. Then there are arbitrarily long words in
$U$. Since each infinite letter from $h^\omega(a)$ produces an
occurrence of $b$ in at most $|\abet|$ steps, there must be
arbitrarily long words from $\abet_F$ in $U$. This is possible only if
for some $c \in \abet_I$ and integer $s \leq |\abet|$ we have
$h^s(c)=v_1cv_2$, where  for $i=1$ or $i=2$ we have $v_i \in
\abet_F^+$ and $h^n(v_i) \neq \varepsilon$ for every $n \geq 0$. This
is a property that we can effectively check. Note that if
$h^n(v_i)=\varepsilon$ for some $n \geq 0$, then
$h^{|\abet|}(v_i)=\varepsilon$. On the other hand, if there exists $c
\in \abet_I$ satisfying the above conditions, the set $U$ is clearly
infinite. Hence, the finiteness of $U$ can be verified and the finite set
$U$ can be effectively constructed.

Now assume that $h^\omega(a)$ is ultimately periodic, i.e.,
$h^\omega(a)=uv^\omega$, where $v$ is primitive. Consider a subset
$U'$ of $U$ containing the elements $u_i$ occurring infinitely many
times in $h^\omega(a)$. Since $b$ is in $\abet_I$, there exists an
integer $N$ such that $|h^n(b)| \geq |v|$ for every $n \geq N$.
Hence, let $n \geq N$. Since $bu_i$ with $u_i \in U'$ occurs in the
periodic part of the sequence, we conclude that $h^n(bu_i) \in
w_n\abet^*$, where $w_n$ is a conjugate of~$v$. Moreover, by the
primitivity of $v$ and $w_n$, we have
\begin{equation}\label{eqWn}
h^n(bu_i) \in w_n^* \quad \text{for all $u_i \in U'$}.
\end{equation}
Namely, assume that $h^n(bu_i)=w_n^t w'$, where $t$ is some positive integer and $w'$ is a proper
prefix of $w_n$, i.e., $w'$ is non-empty and $w' \neq w_n$. Then $h^n(bu_ib) \in
w_n^t w'w_n \abet^*$ is a prefix of $w_n^\omega$, which implies
that the word $w_n$ occurring after $w'$ occurs inside $w_n^2$.
Since $w_n$ is primitive, this is impossible.

Take now any two words $u_i$ and $u_j \in U'$. By~\eqref{eqWn},
we conclude that there exists $m$ such that
$h^\ell(bu_ibu_j)=h^\ell(bu_jbu_i)$ for all $\ell \geq m$. Moreover, by
Theorem~\ref{thCancellation}, we know that we may choose
$m=|\abet|$. Note that if the above does not
hold for some $u_i$ and $u_j$ in $U'$, then $h^\omega(a)$ cannot be
ultimately periodic. Hence, let $m=|\abet|$ and
\[
h^m(bu_ibu_j)=h^m(bu_jbu_i),
\]
for every $u_i, u_j \in U'$. Then the words $h^m(bu_i)$ and
$h^m(bu_j)$ commute and by transitivity we can find a primitive
word $z$ such that
\[
h^\ell(bu_i) \in z^* \quad \text{for all $u_i \in U', \ \ell \geq m$.}
\]
This implies that $h^\omega(a)$ is ultimately $|z|$-periodic. Since
we can test the ultimate $|z|$-periodicity of $h^\omega(a)$ by
Theorem~\ref{thPperiodic}, the ultimate periodicity problem of
$h^\omega(a)$ is decidable.
\end{proof}

\bibliographystyle{eptcs}
\bibliography{D0Lbib}

\end{document}